\documentclass[10pt,english,onecolumn]{IEEEtran}
\usepackage{mathptmx}
\usepackage[T1]{fontenc}
\usepackage[latin9]{inputenc}
\usepackage{geometry}
\usepackage{units}
\usepackage{amsmath}
\usepackage{graphicx}
\usepackage{amssymb}

\newtheorem{example}{Example}
\newtheorem{condition}{Condition}
\newtheorem{definitn}{Definition}
\newtheorem{prop}{Proposition}
\newtheorem{thm}{Theorem}
\newtheorem{remrk}{Remark}

\usepackage{wrapfig}
\usepackage{tikz}  
\usetikzlibrary{plotmarks}
\DeclareMathAlphabet{\mathcal}{OMS}{cmsy}{m}{n}

\usepackage{babel}

\begin{document}

\title{A Rate-Distortion Perspective on Multiple Decoding Attempts for Reed-Solomon
Codes}

\author{Phong S. Nguyen, Henry D. Pfister, and Krishna R. Narayanan\\
{\normalsize Department of Electrical and Computer Engineering,
Texas A\&M University }\\
{\normalsize College Station, TX 77840, U.S.A. }\\
{\normalsize \{psn, hpfister, krn\}@tamu.edu }}
\maketitle
\begin{abstract}
Recently, a number of authors have proposed decoding schemes for Reed-Solomon
(RS) codes based on multiple trials of a simple RS decoding algorithm.
In this paper, we present a rate-distortion (R-D) approach to analyze
these multiple-decoding algorithms for RS codes. This approach is
first used to understand the asymptotic performance-versus-complexity
trade-off of multiple error-and-erasure decoding of RS codes. By defining
an appropriate distortion measure between an error pattern and an
erasure pattern, the condition for a single error-and-erasure decoding
to succeed reduces to a form where the distortion is compared to a
fixed threshold. Finding the best set of erasure patterns for multiple
decoding trials then turns out to be a covering problem which can
be solved asymptotically by rate-distortion theory. Next, this approach
is extended to analyze multiple algebraic soft-decision (ASD) decoding
of RS codes. Both analytical and numerical computations of the R-D
functions for the corresponding distortion measures are discussed.
Simulation results show that proposed algorithms using this approach
perform better than other algorithms with the same complexity.
\end{abstract}
\thispagestyle{empty}\pagestyle{empty}

\global\long\def\half{}

\section{Introduction}

Reed-Solomon (RS) codes are one of the most widely used error-correcting
codes in digital communication and data storage systems. This is primarily
due to the fact that RS codes are maximum distance separable (MDS)
codes, can correct long bursts of errors, and have efficient hard-decision
decoding (HDD) algorithms, such as the Berlekamp-Massey (BM) algorithm,
which can correct up to half the minimum distance ($d_{min})$ of
the code. An~$(n,k)$~RS code of length~$n$~and dimension~$k$
is known to have~$d_{min}=n-k+1$~due to its MDS nature.

Since the arrival of RS codes, people have put a considerable effort
into improving the decoding performance at the expense of complexity.
A breakthrough result of Guruswami and Sudan (GS) introduces a hard-decision
list-decoding algorithm based on algebraic bivariate interpolation
and factorization techniques that can correct errors beyond half the
minimum distance of the code \cite{Guruswami-it99}. Nevertheless,
HDD algorithms do not fully exploit the information provided by the
channel output. Koetter and Vardy (KV) later extended the GS decoder
to an algebraic soft-decision (ASD) decoding algorithm by converting
the probabilities observed at the channel output into algebraic interpolation
conditions in terms of a multiplicity matrix \cite{Koetter-it03}.
Both of these algorithms however have significant computational complexity.
Thus, multiple runs of error-and-erasure and error-only decoding with
some low complexity algorithm, such as the BM algorithm, has renewed
the interest of researchers. These algorithms essentially first construct
a set of either erasure patterns \cite{Forney-it66,Lee-globecom08},
test patterns \cite{Chase-it72}, or patterns combining both \cite{Tang-comlett01}
and then attempt to decode using each pattern. There has also been
recent interest in lowering the complexity per decoding trial as can
be seen in \cite{Bellorado-isit06,Xia-com07,Xia-icc08}.

In the scope of multiple error-and-erasure decoding, there have been
several algorithms using different sets of erasure patterns. After
multiple decoding trials, these algorithms produce a list of candidate
codewords and then pick the best codeword on this list, whose size
is usually small. The nature of multiple error-and-erasure decoding
is to erase some of the least reliable symbols since those symbols
are more prone to be erroneous. The first algorithm of this type is
called Generalized Minimum Distance (GMD) \cite{Forney-it66} and
it repeats error-and-erasure decoding while successively erasing an
even number of the least reliable positions (LRPs) (assuming that~$d_{min}$~is
odd). More recent work by Lee and Kumar \cite{Lee-globecom08} proposes
a soft-information successive (multiple) error-and-erasure decoding
(SED) that achieves better performance but also increases the number
of decoding attempts. Literally, the Lee-Kumar's SED$(l,f)$~algorithm
runs multiple error-and-erasure decoding trials with every combination
of an even number~$\leq f$~of erasures within the~$l$~LRPs. 

A natural question that arises is how to construct the {}``best''
set of erasure patterns for multiple error-and-erasure decoding. Inspired
by this, we first design a rate-distortion framework to analyze the
asymptotic trade-off between performance and complexity of multiple
error-and-erasure decoding of RS codes. The framework is also extended
to analyze multiple algebraic soft-decision decoding (ASD). Next,
we proposed a group of multiple-decoding algorithms based on this
approach that achieve better performance-versus-complexity trade-off
than other algorithms. The multiple-decoding algorithm that achieves
the best trade-off turns out to be a multiple error-only decoding
using the set of patterns generated by random codes combining with
covering codes. These are the main results of this paper.

\subsection{Outline of the paper}

The paper is organized as follows. In Section \ref{sec:MultipleBMA},
we design an appropriate distortion measure and present a rate-distortion
framework to analyze the performance-versus-complexity trade-off of
multiple error-and-erasure decoding of RS codes. Also in this section,
we propose a general multiple-decoding algorithm that can be applied
to error-and-erasure decoding. Then, in Section \ref{sec:Computing-RD},
we discuss a numerical computation of R-D function which is needed
for the proposed algorithm. In Section \ref{sec:Multiple ASD}, we
analyze both bit-level and symbol-level ASD decoding and design distortion
measures so that they can fit into the general algorithm. In Section
\ref{sec:Ext-and-Gen}, we offer some extensions that help the algorithm
achieve better performance and running time. Simulation results are
presented in Section \ref{sec:Simulation-results} and finally, conclusion
is provided in Section \ref{sec:Conclusion}.

\section{Multiple Error-and-Erasure Decoding\label{sec:MultipleBMA}}

In this section, we set up a rate-distortion framework to analyze
multiple attempts of conventional hard decision error-and-erasure
decoding.

Let~$\mathbb{F}_{2^{q}}$~be the Galois field with~$2^{q}$~elements
denoted as~$\alpha_{1},\alpha_{2},\ldots,\alpha_{2^{q}}$. We consider
an $(n,k)$~RS code of length~$n$, dimension~$k$ over~$\mathbb{F}_{2^{q}}$.
Assume that we transmit a codeword~$\mathbf{c}=(c_{1},c_{2},\ldots,c_{n})\in\mathbb{F}_{2^{q}}^{n}$~over
some channel and receive a vector~$\mathbf{r}=(r_{1},r_{2},\ldots,r_{n})\in\mathcal{Y}^{n}$~where~$\mathcal{Y}$~is
the receive alphabet for a single RS symbol. In this paper, we assume
that $\mathcal{Y}=\mathbb{R}^{q}$ and all simulations are based on
transmitting each of the $q$ bits in a symbol using Binary Phase-Shift
Keying (BPSK) on an Additive White Gaussian Noise (AWGN) channel.
For each codeword index~$i$, let~$\pi_{i}:\{1,2,\ldots,2^{q}\}\rightarrow\{1,2,\ldots,2^{q}\}$~be
the permutation given by sorting~$p_{i,j}=\Pr(c_{i}=\alpha_{j}|\mathbf{r})$~in
decreasing order so that~$p_{i,\pi_{i}(1)}\geq p_{i,\pi_{i}(2)}\geq\ldots\geq p_{i,\pi_{i}(2^{q})}$.
Then, we can specify~$y_{i,j}=\alpha_{\pi_{i}(j)}$~as the~$j$-th
most reliable symbol for~$j=1,\ldots,2^{q}$~at codeword index~$i$.
To obtain the reliability of the codeword positions (indices), we
construct the permutation~$\sigma:\{1,2,\ldots,n\}\rightarrow\{1,2,\ldots,n\}$~given
by sorting the probabilities~$p_{i,\pi_{i}(1)}$~of the most likely
symbols in increasing order. Thus, codeword position~$\sigma(i)$~is
the~$i$-th LRP. These above notations will be used throughout this
paper.
\begin{example}
Consider~$n=3$~and~$q=2$. Assume that we have the probability~$p_{i,j}$
written in a matrix form as follows\[
\mathbf{P}=\left(\begin{array}{ccc}
0.01 & 0.01 & \mathbf{0.93}\\
\mathbf{0.94} & 0.03 & 0.04\\
0.03 & \mathbf{0.49} & 0.01\\
0.02 & 0.47 & 0.02\end{array}\right)\,\,\mbox{\mbox{where}}\,\, p_{i,j}=[\mathbf{P}]_{j,i}\]

then~$\pi_{1}(1,2,3,4)=(2,3,4,1),\pi_{2}(1,2,3,4)=(3,4,2,1),\pi_{3}(1,2,3,4)=(1,2,4,3)$~and~$\sigma(1,2,3)=(2,3,1)$.\end{example}
\begin{condition}
\label{con:BMAerr-n-era}(Classical decoding threshold, see \cite{Lin-1983,Blahut-2003}):
If~$e$~symbols are erased, a conventional hard-decision error-and-erasure
decoder such as the BM algorithm is able to correct~$\nu$~errors
in unerased positions if \begin{equation}
2\nu+e<n-k+1.\label{eq:scbma_0}\end{equation}

\end{condition}

\subsection{Conventional error and erasure patterns.}
\begin{definitn}
\label{def:(Conv. patterns)}(Conventional error and erasure patterns)
We define~$x^{n}\in\mathbb{Z}_{2}^{n}\triangleq\{0,1\}^{n}$~and~$\hat{x}^{n}\in\mathbb{Z}_{2}^{n}$~as
an error pattern and an erasure pattern respectively, where~$x_{i}=0$~means
that an error occurs (i.e. the most likely symbol is incorrect) and~$\hat{x}_{i}=0$~means
that an erasure occurs at index~$i$. \end{definitn}
\begin{example}
If~$d_{min}$~is odd then $\{111111\ldots,001111\ldots,000011\ldots,\ldots\}$
is the set of erasure patterns for the GMD algorithm. For the SED$(3,2)$~algorithm,
the set of erasure patterns has the form~$\{111111\ldots,001111\ldots,010111\ldots,100111\ldots\}$.
Here, in each erasure pattern the letters are written in increasing
reliability order of the codeword positions.
\end{example}
Let us revisit the question how to construct the best set of erasure
patterns for multiple error-and-erasure decoding. First, it can be
seen that a multiple error-and-erasure decoding succeeds if the condition
(\ref{eq:scbma_0}) is satisfied during at least one round of decoding.
Thus, our approach is to design a distortion measure that converts
the condition (\ref{eq:scbma_0}) into a form where the distortion
between an error pattern~$x^{n}$~and an erasure pattern~$\hat{x}^{n}$,
denoted as~$d(x^{n},\hat{x}^{n})$, is less than a fixed threshold.
\begin{definitn}
Given a \emph{letter-by-letter }distortion measure~$\delta$, the
distortion between an error pattern~$x^{n}$~and an erasure pattern~$\hat{x}^{n}$~is
defined by\[
d(x^{n},\hat{x}^{n})=\sum_{i=1}^{n}\delta(x_{i},\hat{x}_{i}).\]
\end{definitn}
\begin{prop}
\label{prop:bma1}If we choose the \emph{letter-by-letter} distortion
measure~$\delta:\mathbb{Z}_{2}\times\mathbb{Z}_{2}\rightarrow\mathbb{R}_{\geq0}$~
as follows\begin{equation}
\begin{array}{cc}
\delta(0,0)=1 & \delta(0,1)=2\\
\delta(1,0)=1 & \delta(1,1)=0\end{array}\label{eq:dstfnBMA}\end{equation}
then the condition (\ref{eq:scbma_0}) for a successful error-and-erasure
decoding then reduces to the form where the distortion is less than
a fixed threshold\[
d(x^{n},\hat{x}^{n})<n-k+1.\]
\end{prop}
\begin{proof}
First, we define $\chi_{s,t}\triangleq\left|\{i\in\left\{ 1,2,\ldots,n\right\} :x_{i}=s,\hat{x}_{i}=t\}\right|$
to count the number of~$(x_{i},\hat{x}_{i})$~pairs equal to~$(s,t)$~for
every~$s,t\in\{0,1\}$. Noticing that~$e=\chi_{0,0}+\chi_{1,0}$~and~$\nu=\chi_{0,1}$,
the condition (\ref{eq:scbma_0}) for one error-and-erasure decoding
attempt to succeed becomes ~$2\chi_{0,1}+\chi_{0,0}+\chi_{1,0}<n-k+1$.
By seeing that~$d(x^{n},\hat{x}^{n})=2\chi_{0,1}+\chi_{0,0}+\chi_{1,0}$~we
conclude the proof. \vspace{2mm}
\end{proof}
Next, we try to maximize the chance that this successful decoding
condition is satisfied by at least one of the decoding attempts (i.e.~$d(x^{n},\hat{x}^{n})<n-k+1$~for
at least one erasure patterns~$\hat{x}^{n}$). Mathematically, we
want to build a set~$\mathcal{B}$~of no more than~$2^{R}$~erasure
patterns~$\hat{x}^{n}$~in order to \begin{equation}
\max_{\mathcal{B}:|\mathcal{B}|\leq2^{R}}\Pr\{\min_{\hat{x}^{n}\in\mathcal{B}}d(x^{n},\hat{x}^{n})<n-k+1\}.\label{eq:ProbStatement}\end{equation}
The exact answer to this problem is difficult to find. However, one
can see it as a covering problem where one wants to cover the space
of error patterns using a minimum number of balls centered at the
chosen erasure patterns. 

This view leads to an asymptotic solution of the problem based on
rate-distortion theory. More precisely, we view the error pattern~$x^{n}$~as
a source sequence and the erasure pattern~$\hat{x}^{n}$~as a reproduction
sequence. \begin{wrapfigure}{r}{7.2cm}
\begin{tikzpicture}[scale=1.5]
\tikzstyle{mybox} = [draw=black, fill=blue!05,rectangle,rounded corners]
\fill [green!5] (0.5,0.2) circle (15pt); \fill [green!5] (0,1) circle (15pt); \fill [green!5] (0.9,0.9) circle (15pt); \fill [green!5] (1.4,0.2) circle (15pt); \fill [green!5] (1.8,1) circle (15pt);
\fill [blue] (0.5,0.2) circle (1pt); \draw[thick] (0.5,0.2) circle (15pt);
\fill [blue] (0,1) circle (1pt); \draw[thick] (0,1) circle (15pt);
 \fill [blue] (0.9,0.9) circle (1pt);  \draw[thick] (0.9,0.9) circle (15pt);
 \fill [blue] (1.4,0.2) circle (1pt);  \draw[thick] (1.4,0.2) circle (15pt);
 \fill [blue] (1.8,1) circle (1pt);  \draw[thick] (1.8,1) circle (15pt);
\filldraw [fill=red!40] (1,-0.05) rectangle (1.05,0); \filldraw [fill=red!40] (0.05,-0.05) rectangle (0.1,0); \filldraw [fill=red!40] (-0.35,1.35) rectangle (-0.3,1.4); \filldraw [fill=red!40] (2.2,1.2) rectangle (2.25,1.25); \filldraw [fill=red!40] (1.7,-0.2) rectangle (1.75,-0.15); \filldraw [fill=red!40] (1,1.35) rectangle (1.05,1.4); \filldraw [fill=red!40] (2,1.25) rectangle (2.05,1.3); \filldraw [fill=red!40] (0.7,0.55) rectangle (0.75,0.6); \filldraw [fill=red!40] (1.25,0.65) rectangle (1.3,0.7); \filldraw [fill=red!40] (1.05,0.5) rectangle (1.1,0.55); \filldraw [fill=red!40] (-0.15,0.55) rectangle (-0.1,0.6); \filldraw [fill=red!40] (0.3,1.325) rectangle (0.35,1.375); \filldraw [fill=red!40] (0.45,1.05) rectangle (0.5,1.1); \filldraw [fill=red!40] (0.7,-0.25) rectangle (0.75,-0.2); \filldraw [fill=red!40] (1.85,0.1) rectangle (1.9,0.15); \filldraw [fill=red!40] (2.1,0.7) rectangle (2.15,0.75);
\node [mybox] (box) at (3.2,0) {    \begin{minipage}[t!]{0.15\textwidth}         Error pattern\\         Erasure pattern     \end{minipage}     };
\filldraw [fill=red!40] (3.875,0.125) rectangle (3.925,0.175); \fill [blue] (3.9,-0.125) circle (1pt);
\end{tikzpicture}
\end{wrapfigure} Rate-distortion theory shows that the set~$\mathcal{B}$~of~$2^{R}$~reproduction
sequences can be generated randomly so that\[
\lim_{n\rightarrow\infty}E_{\hat{x}^{n}\in\mathcal{B}}\left[d(x^{n},\hat{x}^{n})\right]\leq D\]
where the distortion~$D$~is minimized for a given rate~$R$. Thus,
for large enough~$n$, we have\[
\min_{\hat{x}^{n}\in\mathcal{B}}d(x^{n},\hat{x}^{n})\leq D\]
with high probability. Here,~$R$~and~$D$~are closely related
to the complexity and the performance, respectively, of the decoding
algorithm. Therefore, we characterize the trade-off between those
two aspects using the relationship between~$R$~and~$D$.

\subsection{Generalized error and erasure patterns}

In this subsection, we consider a generalization of the conventional
error and erasure patterns under the same framework to make better
use of the soft information. At each index of the RS codeword, beside
erasing the symbol we can try to decode using not only the most likely
symbol but also other ones as the hard decision (HD) symbol. To handle
up to the~$l$~most likely symbols at each index~$i$, we let~$\mathbb{Z}_{l+1}\triangleq\{0,1,\ldots,l\}$~and
consider the following definition.
\begin{definitn}
\label{def:(PatternsBMA)}(Generalized error patterns and erasure
patterns) Consider a positive integer~$l<2^{q}$. Let us define $x^{n}\in\mathbb{Z}_{l+1}^{n}$~as
the generalized error pattern where, at index~$i$,~$x_{i}=j$~implies
that the~$j$-th most likely symbol is correct for~$j\in\{1,2,\ldots l\}$,
and~$x_{i}=0$~implies none of the first~$l$~most likely symbols
is correct. Let $\hat{x}^{n}\in\mathbb{Z}_{l+1}^{n}$~be the generalized
erasure pattern used for decoding where, at index~$i$,~$\hat{x}_{i}=j$
implies that the~$j$-th most likely symbol is used as the hard-decision
symbol for~$j\in\{1,2,\ldots,l\}$, and~$\hat{x}_{i}=0$~implies
that an erasure is used at that index. 

For simplicity, we will refer to~$x^{n}$~as the error pattern and~$\hat{x}^{n}$~as
the erasure pattern like in the conventional case. Next, we also want
to convert the condition (\ref{eq:scbma_0}) to the form where~$d(x^{n},\hat{x}^{n})$~is
less than a fixed threshold. Proposition \ref{prop:bma1}\emph{ }is
thereby generalized into the following theorem.\end{definitn}
\begin{thm}
\label{thm:ExtBMA-1}We choose the \emph{letter-by-letter} distortion
measure~$\delta:\mathbb{Z}_{l+1}\times\mathbb{Z}_{l+1}\rightarrow\mathbb{R}_{\geq0}$~
defined by~$\delta(x,\hat{x})=[\Delta]_{x,\hat{x}}$~in terms of
the $(l+1)\times(l+1$)~matrix\[
\Delta=\left(\begin{array}{ccccc}
1 & 2 & \ldots & 2 & 2\\
1 & 0 & \ldots & 2 & 2\\
\vdots & \vdots & \ddots & \vdots & \vdots\\
1 & 2 & \ldots & 0 & 2\\
1 & 2 & \ldots & 2 & 0\end{array}\right).\]
Using this, the condition (\ref{eq:scbma_0}) for a successful error-and-erasure
decoding becomes\[
d(x^{n},\hat{x}^{n})<n-k+1.\]
\end{thm}
\begin{proof}
The reasoning is similar to Proposition \ref{prop:bma1} using the
fact that~$e=\sum_{s=0}^{l}\chi_{s,0}$~and~$\nu=\sum_{t=1,}^{l}\sum_{s=0,s\neq t}^{l}\chi_{s,t}$~where
$\chi_{s,t}\triangleq\left|\{i\in\left\{ 1,2,\ldots,n\right\} :x_{i}=s,\hat{x}_{i}=t\right|$
for every~$s,t\in\mathbb{Z}_{l+1}$.
\end{proof}
For each~$l=1,2,\ldots,2^{q}$, we will refer to this generalized
case as~mBM-$l$~decoding.
\begin{example}
We consider the case mBM-2 decoding where~$l=2$. The distortion
measure is given by following the matrix\[
\Delta=\left(\begin{array}{ccc}
1 & 2 & 2\\
1 & 0 & 2\\
1 & 2 & 2\end{array}\right).\]
Here, at each codeword position, we consider the first and second
most likely symbols as the two hard-decision choices like in the Chase-type
decoding method proposed by Bellorado and Kavcic \cite{Bellorado-isit06}.
\end{example}

\subsection{Proposed General Multiple-Decoding Algorithm \label{sec:Proposed-Algorithm}}

In this section, we propose a general multiple-decoding algorithm
for RS codes based on the rate-distortion approach. This general algorithm
applies to not only multiple error-and-erasure decoding but also multiple-decoding
of other decoding schemes that we will discuss later. The first step
is designing a distortion measure that converts the condition for
a single decoding to succeed to the form where distortion is less
than a fixed threshold. After that, decoding proceeds as described
below.
\begin{itemize}
\item \emph{Phase I: Compute rate-distortion function.}
\end{itemize}
\emph{~~~\,Step 1:} Transmit~$\tau$~(say~$\tau=1000$) arbitrary
test RS codewords, indexed by time~$t=1,2,\ldots,\tau$, over the
channel and compute a set of~$\tau$~$2^{q}\times n$~matrices~$\mathbf{P}_{1}^{(t)}$~where~$[\mathbf{P}_{1}^{(t)}]_{j,i}=p_{i,\pi_{i}^{(t)}(j)}^{(t)}$~is
the probability of the~$j$-th most likely symbol at position~$i$~during
time~$t$.

\emph{Step 2:} For each time~$t$, obtain the matrix~$\mathbf{P}_{2}^{(t)}$~from~$\mathbf{P}_{1}^{(t)}$~through
a permutation~$\sigma^{(t)}:\{1,2,\ldots,n\}\rightarrow\{1,2,\dots,n\}$~that
sorts the probabilities~$p_{i,\pi_{i}^{(t)}(1)}^{(t)}$~in increasing
order to indicate the reliability order of codeword positions. Take
the entry-wise average of all~$\tau$~matrices~$\mathbf{P}_{2}^{(t)}$~to
get an average matrix~$\mathbf{\bar{P}}$.

\emph{Step 3:} Compute the R-D function of a source sequence (error
pattern) with probability of source letters derived from~$\mathbf{\bar{P}}$~and
the designed distortion measure (see Section \ref{sec:Computing-RD}
and Section \ref{sub:AnalyticalRD}) . Determine the point on the
R-D curve that corresponds to a designated rate~$R$~along with
the test-channel input-probability distribution vector~$\underline{q}$~that
achieves that point.
\begin{itemize}
\item \emph{Phase II: Run actual decoder.}
\end{itemize}
\emph{~~\,Step 4: }Based on the actual received signal sequence,
compute~$p_{i,\pi_{i}(1)}$~and determine the permutation~$\sigma$~that
gives the reliability order of codeword positions by sorting~$p_{i,\pi_{i}(1)}$~in
increasing order.

\emph{Step 5:} Randomly generate a set of~$2^{R}$~erasure patterns
using the test-channel input-probability distribution vector~$\underline{q}$
and permute the indices of each erasure pattern by the permutation~$\sigma^{-1}.$

\emph{Step 6:} Run multiple attempts of the corresponding decoding
scheme (e.g. error-and erasure decoding) using the set of erasure
patterns in Step 5 to produce a list of candidate codewords.

\emph{Step 7:} Use Maximum-Likelihood (ML) decoding to pick the best
codeword on the list.

\section{Computing The Rate-Distortion Function\label{sec:Computing-RD}}

In this section, we will present a numerical method to compute the
R-D function and the test-channel input-probability distribution that
achieves a specific point in the R-D curve. This probability distribution
will be needed to randomly generate the set of erasure patterns in
the general multiple-decoding algorithm that we have proposed.

For an arbitrary discrete distortion measure, it can be difficult
to compute the R-D function analytically. Fortunately, the Blahut-Arimoto
(B-A) algorithm (see details in \cite{Blahut-it72,Arimoto-it72})
gives an alternating minimization technique that efficiently computes
the R-D function of a single discrete source. More precisely, given
a parameter~$s<0$~which represents the slope of the~$R-D$~curve
at a specific point and an arbitrary all-positive initial test-channel
input-probability distribution vector~$\underline{q}^{(0)}$, the
B-A algorithm shows us how to compute the rate-distortion point~$(R_{s},D_{s})$~by
means of computing the test-channel input-probability distribution
vector~$\underline{q}^{\star}=\lim_{t\rightarrow\infty}\underline{q}^{(t)}$
and the test-channel transition probability matrix~$Q^{\star}=\lim_{t\rightarrow\infty}Q^{(t)}$~that
achieves that point.

However, it is not straightforward to apply the B-A algorithm to compute
the R-D for a discrete source sequence~$x^{n}$~(an error pattern
in our context) of~$n$~independent but non identical source components~$x_{i}$.
In order to do that, we consider the group of source letters~$(j_{1},j_{2},\ldots,j_{n})$~where~$j_{i}\in\mathcal{J}$~as
a super-source letter~$J\in\mathcal{J}^{n}$, the group of reproduction
letters~$(k_{1},k_{2},\ldots,k_{n})$~where~$k_{i}\in\mathcal{K}$~as
a super-reproduction letter~$K\in\mathcal{K}^{n}$, and the source
sequence~$x^{n}$~as a single source. For each super-source letter~$J$,~$p_{J}=\Pr(x^{n}=J)=\prod_{i=1}^{n}\Pr(x_{i}=j_{i})=\prod_{i=1}^{n}p_{j_{i}}$~follows
from the independence of source components. While we could apply the
B-A algorithm to this source directly, the complexity is a problem
because the alphabet sizes for~$J$~and~$K$~ become the super-alphabet
sizes~$|\mathcal{J}|^{n}$~and~$|\mathcal{K}|^{n}$~respectively.
Instead, we avoid this computational challenge by choosing the initial
test-channel input-probability distribution so that it can be factorized
into a product of~$n$~initial test-channel input-probability components,
i.e.~$q_{K}^{(0)}=\prod_{i=1}q_{k_{i}}^{(0)}$. Then, we see that
this factorization rule still applies after every step of the iterative
process. By doing this, for each parameter~$s$~we only need to
compute the rate-distortion pair for each component (or index~$i$)
separately and sum them together. This is captured into the following
theorem.
\begin{thm}
\label{thm:(Factored-Blahut)}(Factored Blahut-Arimoto algorithm)
Consider a discrete source sequence~$x^{n}$~of~$n$~independent
but non identical source components~$x_{i}$. Given a parameter~$s<0$,
the rate and the distortion for this source sequence are given by\[
R_{s}=\sum_{i=1}^{n}R_{i,s}\,\,\mbox{and}\,\, D_{s}=\sum_{i=1}^{n}D_{i,s}\]
where the components~$R_{i,s}$~~and~$D_{i,s}$~are computed
by the B-A algorithm with the parameter~$s$. This pair of rate and
distortion can be achieved by the corresponding test-channel input-probability
distribution~$q_{K}\triangleq\Pr(\hat{x}^{n}=K)=\prod_{i=1}^{n}q_{k_{i}}$~where
the component probability distribution~$q_{k_{i}}\triangleq\Pr(\hat{x}_{i}=k_{i})$.\end{thm}
\begin{proof}
See Appendix \ref{sec:AppBlahut}.
\end{proof}

\section{Multiple Algebraic Soft Decision Decoding (ASD)\label{sec:Multiple ASD}}

In this section, we analyze and design a distortion measure to convert
the condition for successful ASD decoding to a suitable form so that
we can apply the general multiple-decoding algorithm to ASD decoding.

First, let us give a brief review on ASD decoding of RS codes. Given
a set~$\{\beta_{1},\beta_{2},\ldots,\beta_{n}\}$~of~$n$~distinct
elements in~$\mathbb{F}_{2^{q}}.$ From each message polynomial~$f(X)=f_{0}+f_{1}X+\ldots+f_{k-1}X^{k-1}$,
we can have a codeword~$c=(c_{1},c_{2},\ldots,c_{n})$~by evaluating
the message polynomial at~$\{\beta_{i}\}_{i=1}^{n}$, i.e.~$c_{i}=f(\beta_{i})$~for~$i=1,2,\ldots,n$.
Consider a received vector~$\mathbf{r}=(r_{1},r_{2},\ldots,r_{n})$,
we can compute the \emph{a posteriori} probability (APP) matrix~$\mathbf{P}$~as
follows.\[
[\mathbf{P}]_{j,i}=p_{i,j}=\Pr(c_{i}=\alpha_{j}|\mathbf{r})\,\,\,\mbox{for}\,\,1\leq i\leq n,1\leq j\leq2^{q}.\]
The ASD decoding as in \cite{Koetter-it03} has the following main
steps.
\begin{enumerate}
\item \emph{Multiplicity Assignment}: Use a particular multiplicity assignment
scheme (MAS) to derive a~$2^{q}\times n$ multiplicity matrix, denoted
as~$\mathbf{M}$, of non-negative integer entries~$\{m_{i,j}\}$~from
the APP matrix~$\mathbf{P}$.
\item \emph{Interpolation}: Construct a bivariate polynomial~$Q(X,Y)$~of
minimum~$(1,k-1)$~weighted degree that passes through each of the
point~$(\beta_{j},\alpha_{i})$~with multiplicity~$m_{i,j}$~for~$i=1,2,\ldots,2^{q}$~and~$j=1,2,\ldots,n$.
\item \emph{Factorization}: Find all polynomials~$f(X)$~of degree less
than~$k$~such that~$Y-f(X)$~is a factor of~$Q(X,Y)$~and re-evaluate
these polynomials to form a list of candidate codewords.
\end{enumerate}
In this paper, we denote~$m=\max_{i,j}m_{i,j}$~as the maximum multiplicity.
Intuitively, higher multiplicity should be put on more likely symbols.
Increasing~$m$~generally gives rise to the performance of ASD decoding.
However, one of the drawbacks of ASD decoding is that its decoding
complexity is roughly~$\mathcal{O}(m^{6})$~which sharply increases
with~$m$. Thus, in this section we will work with small~$m$~to
keep the complexity affordable. 

One of the main contributions of \cite{Koetter-it03} is to offer
a condition for successful ASD decoding represented in terms of two
quantities specified as the score and the cost as follows.
\begin{definitn}
The score~$S_{\mathbf{M}}(\mathbf{c})$~with respect to a codeword~$\mathbf{c}$~and
a multiplicity matrix~$\mathbf{M}$~is defined as~$S_{\mathbf{M}}(\mathbf{c})=\sum_{j=1}^{n}m_{[c_{j}],j}$\\
where~$[c_{j}]=i$~such that~$\alpha_{i}=c_{j}$. The cost~$C_{\mathfrak{\mathbf{M}}}$~of
a multiplicity matrix~$\mathbf{M}$~is defined as~$C_{\mathbf{M}}=\frac{1}{2}\sum_{i=1}^{q}\sum_{j=1}^{n}m_{i,j}(m_{i,j}+1)$\end{definitn}
\begin{condition}
(ASD decoding threshold, see \cite{Koetter-it03,Jiang-it08,McEliece-IPN03}).
The transmitted codeword will be on the list if\begin{eqnarray*}
T(S_{\mathbf{M}}) & > & C_{\mathbf{M}}\,\,\,\mbox{where\ensuremath{\,\,}}T(S_{\mathbf{M}})=(a+1)\left[S_{\mathbf{M}}-\frac{a}{2}(k-1)\right]\end{eqnarray*}
\begin{equation}
\mbox{for any}\,\, a\in\mathbb{N}\,\,\mbox{such that}\,\, a(k-1)<S_{\mathbf{M}}\leq(a+1)(k-1).\label{eq:ASDsc}\end{equation}

\end{condition}
To match the general framework, the ASD decoding threshold (or condition
for successful ASD decoding) should be converted to the form where
the distortion is smaller than a fixed threshold.

\subsection{Bit-level ASD case}

In this subsection, we consider multiple trials of ASD decoding using
bit-level erasure patterns. A bit-level error pattern~$b^{N}\in\mathbb{Z}_{2}^{N}$~and
a bit-level erasure pattern~$\hat{b}^{N}\in\mathbb{Z}_{2}^{N}$~has
length~$N=n\times q$~since each symbol has~$q$~bits. Similar
to Definition \ref{def:(Conv. patterns)} of a conventional error
pattern and a conventional erasure pattern,~$b_{i}=0$~in a bit-level
error pattern implies a bit-level error occurs and~$\hat{b}_{i}$~in
a bit-level erasure pattern implies that a bit-level erasure occurs. 

From each bit-level erasure pattern we can specify entries of the
multiplicity matrix~$\mathbf{M}$~using the bit-level MAS proposed
in \cite{Jiang-it08} as follows: for each codeword position, assign
multiplicity 2 to the symbol with no bit erased, assign multiplicity
1 to each of the two candidate symbols if there is 1 bit erased, and
assign multiplicity zero to all the symbols if there are~$\geq2$~bits
erased. All the other entries are zeros by default. This MAS has a
larger decoding region compared to the conventional error-and-erasure
decoding scheme.
\begin{condition}
(Bit-level ASD decoding threshold, see \cite{Jiang-it08}) For RS
codes of rate~$\frac{k}{n}\geq\frac{2}{3}+\frac{1}{n}$, ASD decoding
using the proposed bit-level MAS will succeed (i.e. the transmitted
codeword is on the list) if\begin{equation}
3\nu_{b}+e_{b}<\frac{1}{2}(n-k+1)\label{eq:bgmdsc}\end{equation}
where~$e_{b}$~is the number of bit-level erasures and~$\nu_{b}$~is
the number of bit-level errors in unerased locations.
\end{condition}
We can choose an appropriate distortion measure according to the following
proposition which is a natural extension of Proposition \ref{prop:bma1}
in the symbol level.
\begin{prop}
\label{prop:BitASD}If we choose the bit-level \emph{letter-by-letter}
distortion measure~$\delta:\mathbb{Z}_{2}\times\mathbb{Z}_{2}\rightarrow\mathbb{R}_{\geq0}$~as
follows\begin{equation}
\begin{array}{cc}
\delta(0,0)=1 & \delta(0,1)=3\\
\delta(1,0)=1 & \delta(1,1)=0\end{array}\label{eq:dstfnBGMD}\end{equation}
then the condition (\ref{eq:bgmdsc}) becomes \begin{equation}
d(b^{N},\hat{b}^{N})<\frac{1}{2}(n-k+1).\label{eq:scbgmd2}\end{equation}
\end{prop}
\begin{proof}
The proof uses the same reasoning as the proof of\emph{ }Proposition\emph{
}\ref{prop:bma1}.\end{proof}
\begin{remrk}
We refer the the multiple-decoding of bit-level ASD as m-b-ASD.
\end{remrk}

\subsection{Symbol-level ASD case}

In this subsection, we try to convert the condition for successful
ASD decoding in general to the form that suits our goal. We will also
determine which multiplicity assignment schemes allow us to do so. 
\begin{definitn}
(Multiplicity type) For some codeword position, let us assign multiplicity~$m_{j}$~to
the~$j$-th most likely symbol for~$j=1,2,\ldots,l$~where~$l\leq2^{q}$.
The remaining entries in the column are zeros by default. We call
the sequence,~$(m_{1},m_{2},\ldots,m_{l})$, the column multiplicity
type for {}``top-$l$'' decoding.
\end{definitn}
First, we notice that a choice of multiplicity types in ASD decoding
at each codeword position has the similar meaning to a choice of erasure
decisions in the conventional error-and-erasure decoding. However,
in ASD decoding we are more flexible and may have more types of erasures.
For example, assigning multiplicity zero to all the symbols (all-zero
multiplicity type) at codeword position~$i$~corresponds to the
case when we have a complete erasure at that position. Assigning the
maximum multiplicity~$m$~to one symbol corresponds to the case
when we choose that symbol as the hard-decision one. Hence with some
abuse of terminology, we also use the term (generalized) erasure pattern~$\hat{x}^{n}$~for
the multiplicity assignment scheme in the ASD context. Each erasure-letter~$x_{i}$~gives
the multiplicity type for the corresponding column of the multiplicity
matrix~$\mathbf{M}$. 
\begin{definitn}
(Error and erasure patterns for ASD decoding) Consider a MAS with~$z$~multiplicity
types. Let~$\hat{x}^{n}\in\{1,2\ldots,z\}^{n}$ be an erasure pattern
where, at index~$i$,~$x_{i}=j$~implies that~multiplicity type~$j$~is
used at column~$i$~of the multiplicity matrix~$\mathbf{M}$. Notice
that the definition of an error pattern~$x^{n}\in\mathbb{Z}_{l+1}^{n}$~in
Definition \ref{def:(PatternsBMA)} applies unchanged here. 
\end{definitn}
Rate-distortion theory gives us the intuition that in general the
more multiplicity types (erasure choices) we have, the better performance
of multiple ASD decoding we achieve as~$n$~becomes large. Thus,
we want to find as many as possible multiplicity types for {}``top-$l$''
that allow us to convert condition for successful ASD decoding to
the correct form.
\begin{example}
Choosing $m=2$, for example, gives four column multiplicity types
for {}``top-2'' decoding as follows: the first is~$(2,0)$~where
we assign multiplicity 2 to the most likely symbol~$y_{i,1}$, the
second is~$(1,1)$~where we assign equal multiplicity 1 to the first
and second most likely symbols~$y_{i,1}$~and~$y_{i,2}$, the third
is~$(0,2)$~where we assign multiplicity 2 to the second most likely
symbol~$y_{i,2}$, and the fourth is~$(0,0)$~where we assign multiplicity
zero to all the symbols at index~$i$~(i.e. the~$i$-th column
of~$\mathbf{M}$~is an all-zero column). As a corollary of Theorem
\ref{thm:GenASD} below, the distortion matrix that converts (\ref{eq:ASDsc})
to the correct form for this case is\[
\Delta=\left(\begin{array}{cccc}
2 & \nicefrac{5}{3} & 2 & 1\\
0 & \nicefrac{2}{3} & 2 & 1\\
2 & \nicefrac{5}{3} & 0 & 1\end{array}\right).\]

\end{example}
The following definition and theorem provide a set of allowable multiplicity
types that converts the condition for successful ASD decoding into
the form where distortion is less than a fixed threshold. 
\begin{definitn}
The set of allowable multiplicity types for {}``top-$l$''~decoding
with maximum multiplicity~$m$ is defined to be%
\footnote{We use the convention that~$\min_{r:m_{r}\neq0}m_{r}=0$~if~$\left\{ r:m_{r}\neq0\right\} =\emptyset$.%
}\begin{equation}
\mathcal{A}(m,l)\triangleq\left\{ (m_{1},m_{2},\ldots,m_{l}):\sum_{r=1}^{l}m_{r}\leq m\,\,\mbox{and}\,\,\sum_{r=1}^{l}m_{r}(m-m_{r})\leq(m+1)\left(\left|\left\{ r:m_{r}\neq0\right\} \right|-1\right)\min_{r:m_{r}\neq0}m_{r}\right\} .\label{eq:AllowableMAS}\end{equation}
Taking the elements of this set in an arbitrary order, we let the
$j$-th multiplicity type in the allowable set be~$(m_{j,1},m_{j,2},\ldots m_{j,l})$. \end{definitn}
\begin{example}
$\mathcal{A}(3,2)$~consists of all permutations of~$(3,0),(2,1),(1,1),(0,0)$.
Meanwhile, $\mathcal{A}(2,2)$~comprises all the permutations of~$(2,0),(1,1),(0,0)$
and we refer to the multiple ASD decoding algorithm using this set
of multiplicity types as mASD-2.~$\mathcal{A}(3,3)$~consists of
all the permutations of~$(3,0,0),(0,0,0),(1,1,0),(2,1,0),(1,1,1)$
and this case is referred as mASD-3. We also consider another case
called~mASD-2a that uses the set of multiplicity types~$\{(2,0),(1,1),(0,0)\}$.\end{example}
\begin{thm}
\label{thm:GenASD} Let $z=\left|\mathcal{A}(m,l)\right|$ be the
number of multiplicity types in a MAS for {}``top-$l$''~decoding
with maximum multiplicity~$m$. Let $\delta:\mathbb{Z}_{l+1}\times\mathbb{Z}_{z+1}\setminus\{0\}\rightarrow\mathbb{R}_{\geq0}$
be a \emph{letter-by-letter} distortion measure defined by~$\delta(x,\hat{x})=[\Delta]_{x,\hat{x}}$,
where $\Delta$is the~$(l+1)\times z$~matrix\[
\Delta=\left(\begin{array}{cccc}
\mu_{1} & \mu_{2} & \ldots & \mu_{z}\\
\mu_{1}-\nicefrac{2m_{1,1}}{m} & \mu_{2}-\nicefrac{2m_{2,1}}{m} & \ldots & \mu_{z}-\nicefrac{2m_{z,1}}{m}\\
\mu_{1}-\nicefrac{2m_{1,2}}{m} & \mu_{2}-\nicefrac{2m_{2,2}}{m} & \ldots & \mu_{z}-\nicefrac{2m_{z,2}}{m}\\
\vdots & \vdots & \ddots & \vdots\\
\mu_{1}-\nicefrac{2m_{1,l}}{m} & \mu_{2}-\nicefrac{2m_{2,l}}{m} & \ldots & \mu_{z}-\nicefrac{2m_{z,l}}{m}\end{array}\right).\]
with~$\mu_{t}=1+\sum_{r=1}^{l}\frac{m_{t,r}(m_{t,r}+1)}{m(m+1)}$.
Then, the condition (\ref{eq:ASDsc}) for successful ASD decoding
of a RS code with rate~$\frac{k}{n}\geq\frac{1}{n}+\frac{m(m+3)}{(m+1)(m+2)}$~is
equivalent to \begin{equation}
d(x^{n},\hat{x}^{n})<n-k+1.\label{eq:scasd1}\end{equation}
\end{thm}
\begin{proof}
[Sketch of proof] (See details in \cite{Nguyen-prep08}) Let $S$
and $C$ be the score and cost of the multiplicity assignment. First,
we show that~$S>\frac{C}{a+1}+\frac{a}{2}(k-1)$~in (\ref{eq:ASDsc})
implies that~$\left(m-\frac{m(m+1)}{2(a+1)}\right)n-\frac{a}{2}(k-1)\geq0$.
Combining this inequality with the high-rate constraint in Theorem
\ref{thm:GenASD} implies that~$a<m+1$. From (\ref{eq:ASDsc}),
we also know that~$(a+1)S-C>\frac{1}{2}a(a+1)(k-1)\geq\frac{1}{2}aS$~and
this implies that~$2C<(a+2)S$. But, the conditions of the theorem
can also be used to show that $2C\geq(m+1)S$. Combining this with
$2C<(a+2)S$ gives a contradiction unless~$a>m-1$. Thus, we conclude
that~$a=m$. 

Therefore, the condition in (\ref{eq:ASDsc}) is equivalent to~$S>\frac{C}{m+1}+\frac{m}{2}(k-1)$~because~$a(k-1)<S$~is
a consequence of~$a=m$ and~$S\leq(m+1)(k-1)$ is satisfied by the
high-rate constraint. Finally, one can show that~$S>\frac{C}{m+1}+\frac{m}{2}(k-1)$~is
equivalent to~$d(x^{n},\hat{x}^{n})<n-k+1$~with the chosen distortion
matrix.\end{proof}
\begin{remrk}
For a fixed~$m$, the size of $\mathcal{A}(m,l)$~is maximized~when~$l=m$.
Multiplicity types~$\smash{(\underbrace{0,\ldots,0}_{m}),(\underbrace{1,\ldots,1}_{m})}$~and
any permutation of~$(m,0,\ldots,0),$$(\lfloor\frac{m}{2}\rfloor,\lfloor\frac{m}{2}\rfloor,0,\ldots,0)$~are
always in the allowable set~$\mathcal{A}(m,m)$.
\end{remrk}

\section{Some Extensions and Generalizations\label{sec:Ext-and-Gen}}

\subsection{Erasure patterns using covering codes}

The R-D framework we use is most suitable when~$n\rightarrow\infty$.
For a finite~$n$, the random coding approach may have problems with
only a few LRPs. We can instead use good covering codes to handle
these LRPs. In the scope of covering problems, one can use an~$l$-ary
$t_{c}$-covering code (e.g. a perfect Hamming or Golay code) with
covering radius~$t_{c}$~to cover the whole space of~$l$-ary vectors
of the same length. The covering may still work well if the distortion
measure is close to, but not exactly equal to the Hamming distortion. 

In order take care of up to the~$l$~most likely symbols at each
of the~$n_{p}$~LRPs of an~$(n,k)$~RS, we consider an~$(n_{c},k_{c})$~$l$-ary~$t_{c}$-covering
code whose codeword alphabet is~$\mathbb{Z}_{l+1}\setminus\{0\}=\{1,2,\ldots,l\}.$
Then, we give a definition of the (generalized) error patterns and
erasure patterns for this case. In order to draw similarities between
this case and the previous cases, we still use the terminology {}``generalized
erasure pattern'' and shorten it to erasure pattern even if error-only
decoding is used. For error-only decoding, Condition \ref{con:BMAerr-n-era}\emph{
}for successful decoding becomes\[
\nu<\frac{1}{2}(n-k+1).\]

\begin{definitn}
(Error and erasure patterns for error-only decoding) Let us define
$x^{n}\in\mathbb{Z}_{l+1}^{n}=\{0,1,\ldots,l\}^{n}$~as an error
pattern where, at index~$i$,~$x_{i}=j$~implies that the~$j$-th
most likely symbol is correct for~$j\in\{1,2,\ldots l\}$, and~$x_{i}=0$~implies
none of the first~$l$~most likely symbols is correct. Let $\hat{x}^{n}\in\{1,2,\ldots,l\}^{n}$~be
an erasure pattern where, at index~$i$,~$\hat{x}_{i}=j$ implies
that the~$j$-th most likely symbol is chosen as the hard-decision
symbol for~$j\in\{1,2,\ldots,l\}$.\end{definitn}
\begin{prop}
If we choose the \emph{letter-by-letter} distortion measure~$\delta:\mathbb{Z}_{l+1}\times\mathbb{Z}_{l+1}\setminus\{0\}\rightarrow\mathbb{R}_{\geq0}$~
defined by~$\delta(x,\hat{x})=[\Delta]_{x,\hat{x}}$~in terms of
the~$(l+1)\times l$~matrix\begin{equation}
\Delta=\left(\begin{array}{cccc}
1 & 1 & \ldots & 1\\
0 & 1 & \ldots & 1\\
1 & 0 & \ldots & 1\\
\vdots & \vdots & \ddots & \vdots\\
1 & 1 & \ldots & 0\end{array}\right)\label{eq:PFdst}\end{equation}
then the condition for successful error-only decoding then becomes\begin{equation}
d(x^{n},\hat{x}^{n})<\frac{1}{2}(n-k+1).\label{eq:SCPF}\end{equation}
\end{prop}
\begin{proof}
It follows directly from~$d(x^{n},\hat{x}^{n})=\sum_{t=1}^{l}\sum_{s=0,s\neq t}^{l}\chi_{s,t}=v$.\end{proof}
\begin{remrk}
If we delete the first row which corresponds to the case where none
of the first~$l$~most likely symbols is correct then the distortion
measure is exactly the Hamming distortion. 
\end{remrk}

\paragraph*{Split covering approach:}

We can break an error pattern~$x^{n}$~into two sub-error patterns~$x^{LRPs}\triangleq x_{\sigma(1)}x_{\sigma(2)}\ldots x_{\sigma(n_{c})}$~of~$n_{c}$
least reliable positions and~$x^{MRPs}\triangleq x_{\sigma(n_{c}+1)}\ldots x_{\sigma(n)}$~of~$n-n_{c}$~most
reliable positions. Similarly, we can break an erasure pattern~$\hat{x}^{n}$~into
two sub-erasure patterns~$\hat{x}^{LRPs}\triangleq\hat{x}_{\sigma(1)}\hat{x}_{\sigma(2)}\ldots\hat{x}_{\sigma(n_{c})}$~and~$\hat{x}^{MRPs}\triangleq\hat{x}_{\sigma(n_{c}+1)}\ldots\hat{x}_{\sigma(n)}$.
Let $z_{n_{c}}$ be the number of positions in the~$n_{c}$~LRPs
where none of the first~$l$~most likely symbols~is correct, or
$z_{n_{c}}=\left|\left\{ i=1,2,\ldots,n_{c}:x_{\sigma(i)}=0\right\} \right|$.
If we assign the set of all sub-error patterns~$\hat{x}^{LRPs}$~to
be an~$(n_{c},k_{c})$~$t_{c}$-covering code then~$d(x^{LRPs},\hat{x}^{LRPs})\leq t_{c}+z_{n_{p}}$~because
this covering code has covering radius~$t_{c}$. Since~$d(x^{n},\hat{x}^{n})=d(x^{LRPs},\hat{x}^{LRPs})+d(x^{MRPs},\hat{x}^{MRPs})$,
in order to increase the probability that the condition (\ref{eq:SCPF})
is satisfied we want to make~$d(x^{MRPs},\hat{x}^{MRPs})$~as small
as possible by the use of the R-D approach. The following proposition
summarizes how to generate a set of~$2^{R}$~erasure patterns for
multiple runs of error-only decoding.
\begin{prop}
In each erasure pattern, the letter sequence at~$n_{c}$~LRPs is
set to be a codeword of an~$(n_{c},k_{c})$~$l$-ary~$t_{c}-$covering
code. The letter sequence of the remaining~$n-n_{c}\,\mbox{MRPs}$
is generated randomly by the R-D method (see Section \ref{sec:Proposed-Algorithm})
with rate~$R_{MRPs}=R-k_{c}\log_{2}l$~and the distortion measure
in (\ref{eq:PFdst}). Since this covering code has~$l^{k_{c}}$~codewords,
the total rate is~$R_{MRPs}+\log_{2}l^{k_{c}}=R.$\end{prop}
\begin{example}
For a (7,4,3) binary Hamming code which has covering radius~$t_{c}=1$,
we take care of the~$2$~most likely symbols at each of the 7 LRPs.
We see that~$1001001$~is a codeword of this Hamming code and then
form erasure patterns~$1001001\hat{x}_{8}\hat{x}_{9}\ldots\hat{x}_{n}$~with
assumption that the positions are written in increasing reliability
order. The~$2^{R-4}$ sub-erasure patterns~$\hat{x}_{8}\hat{x}_{9}\ldots\hat{x}_{n}$~are
generated randomly using the R-D approach with rate~$(R-4)$.\end{example}
\begin{remrk}
While it also makes sense to use a covering codes for the~$n_{c}$~LRPs
of the erasure patterns and set the the rest to be letter~$1$ (i.e.
chose the most likely symbol as the hard-decision), our simulation
results shows that the performance can be improved by using a combination
of covering codes and random (i.e., generated by the R-D approach)
codes.
\end{remrk}

\subsection{Closed form rate-distortion functions\label{sub:AnalyticalRD}}

For some simple distortion measures, we can compute the R-D functions
analytically in closed form. First, we observe an error pattern as
a sequence of independent but non-identical random sources. Then,
we compute the component R-D functions at each index of the sequence
and use convex optimization techniques to allocate the total rate
and distortion to various components. This method converges to the
solution faster than the numerical method in Section \ref{sec:Computing-RD}.
The following two theorems describe how to compute the R-D functions
for the simple distortion measures of Proposition \ref{prop:bma1}
and \ref{prop:BitASD}.
\begin{thm}
\label{thm:(BMA-RD)}(Conventional error-and-erasure decoding) Let
$p_{i}\triangleq\Pr(x_{i}=1)$, the overall rate-distortion function
is given by~$R(D)=\sum_{i=1}^{n}\left[H(p_{i})-H(\tilde{D}_{i})\right]^{+}$~where
$\tilde{D}_{i}\triangleq D_{i}+p_{i}-1$~and~$\tilde{D}_{i}$~can
be found be a reverse water-filling procedure:\[
\tilde{D}_{i}=\begin{cases}
\lambda & \mbox{if}\,\,\lambda<\min\{p_{i},1-p_{i}\}\\
\min\{p_{i},1-p_{i}\} & \mbox{otherwise}\end{cases}\]
where $\lambda$ should be chosen so that~$\sum_{i=1}^{n}\tilde{D}_{i}=D+\sum_{i=1}^{n}p_{i}-n$.
The~$R(D)$~function can be achieved by the test-channel input-probability
distribution\[
q_{0,i}\triangleq\Pr(\hat{x}_{i}=0)=\frac{1-p_{i}-\tilde{D}_{i}}{1-2\tilde{D}_{i}}\,\,\,\mbox{and}\,\,\, q_{1,i}\triangleq\Pr(\hat{x}_{i}=1)=\frac{p_{i}-\tilde{D}_{i}}{1-2\tilde{D}_{i}}.\]
\end{thm}
\begin{proof}
[Sketch of proof] (See \cite{Nguyen-prep08} for details) With the
distortion measure in (\ref{eq:dstfnBMA}), we follow the method in
\cite{Berger-1971} to compute the rate-distortion function component~$R_{i}(D_{i})=\left[H(p_{i})-H(D_{i}+p_{i}-1)\right]^{+}$~and
the test-channel input-probability distribution~$q_{0,i}=\frac{2(1-p_{i})-D_{i}}{3-2(p_{i}+D_{i})}$~and~$q_{1,i}=\frac{1-D_{i}}{3-2(p_{i}+D_{i})}$~for
each index~$i$. Then, one can show that the optimal allocation of
rate and distortion to the various components is given by a reverse-water
filling procedure like in \cite{Cover-1991}.\end{proof}
\begin{thm}
\label{thm:(bASD-RD)}(Bit-level ASD case in Proposition \ref{prop:BitASD})
The overall rate-distortion function in this case is given by~$R(D)=\sum_{i=1}^{N}\left[R_{i}(\lambda)\right]^{+}$~where~$R_{i}(\lambda)=H(p_{i})-H(\frac{1+\lambda}{1+\lambda+\lambda^{2}})+(p_{i}-\frac{1+\lambda}{1+\lambda+\lambda^{2}})H(\frac{\lambda}{1+\lambda})$
~and the distortion component $D_{i}$ is given by\[
D_{i}=\begin{cases}
\frac{1+2\lambda+3\lambda^{2}}{1+\lambda+\lambda^{2}}-p_{i}\frac{1+2\lambda}{1+\lambda} & \mbox{if}\,\, R_{i}(\lambda)>0\\
\min\{1,3(1-p_{i})\} & \mbox{otherwise}\end{cases}\]
where~$\lambda\in(0,1)$~should be chosen so that~$\sum_{i=1}^{N}D_{i}=D$.
The~$R(D)$~function can be achieved by the test-channel input-probability
distribution\[
q_{i,0}\triangleq\Pr(\hat{b}_{i}=0)=\frac{(1+\lambda)-p_{i}(1+\lambda+\lambda^{2})}{1-\lambda^{2}}\,\,\,\mbox{and}\,\,\, q_{i,1}=\Pr(\hat{b}_{i}=1)=\frac{p_{i}(1+\lambda+\lambda^{2})-\lambda(1+\lambda)}{1-\lambda^{2}}.\]
\end{thm}
\begin{proof}
[Sketch of proof] (See \cite{Nguyen-prep08} for details) With the
distortion measure in (\ref{eq:dstfnBGMD}), using the method in \cite{Berger-1971}
we can compute the rate-distortion function component~$R_{i}(\lambda_{i})=H(p_{i})-H(\frac{1+\lambda_{i}}{1+\lambda_{i}+\lambda_{i}^{2}})+(p_{i}-\frac{1+\lambda_{i}}{1+\lambda_{i}+\lambda_{i}^{2}})H(\frac{\lambda_{i}}{1+\lambda_{i}})$~where~$\lambda_{i}$~is
a Lagrange multiplier such that~$D_{i}=\frac{1+2\lambda_{i}+3\lambda_{i}^{2}}{1+\lambda_{i}+\lambda_{i}^{2}}-p_{i}\frac{1+2\lambda_{i}}{1+\lambda_{i}}$~for
each index~$i$. Then, the Kuhn-Tucker conditions define the the
overall rate allocation.
\end{proof}

\section{Simulation results\label{sec:Simulation-results}}

\begin{figure}[t]
\begin{minipage}[c][1\totalheight][t]{0.485\textwidth}%
\includegraphics[width=0.39\paperwidth]{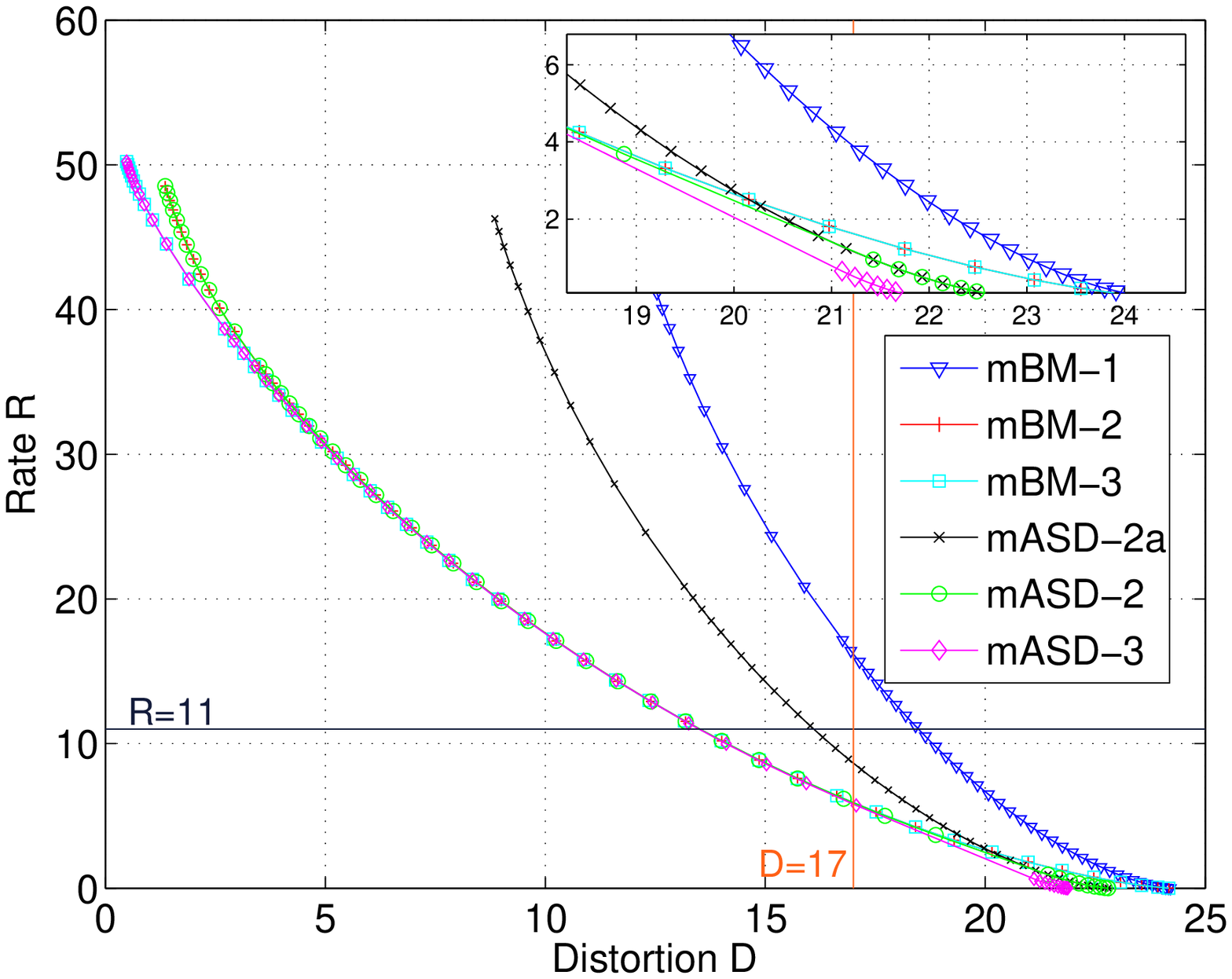}\caption{\label{fig:rdcurve} A realization of R-D curves at~$E_{b}/N_{0}=5.2$dB
for various decoding algorithms for the (255,239) RS code over an
AWGN channel. }
\end{minipage}\hspace{0.5cm}%
\begin{minipage}[c][1\totalheight][t]{0.485\textwidth}%
\includegraphics[width=0.4\paperwidth]{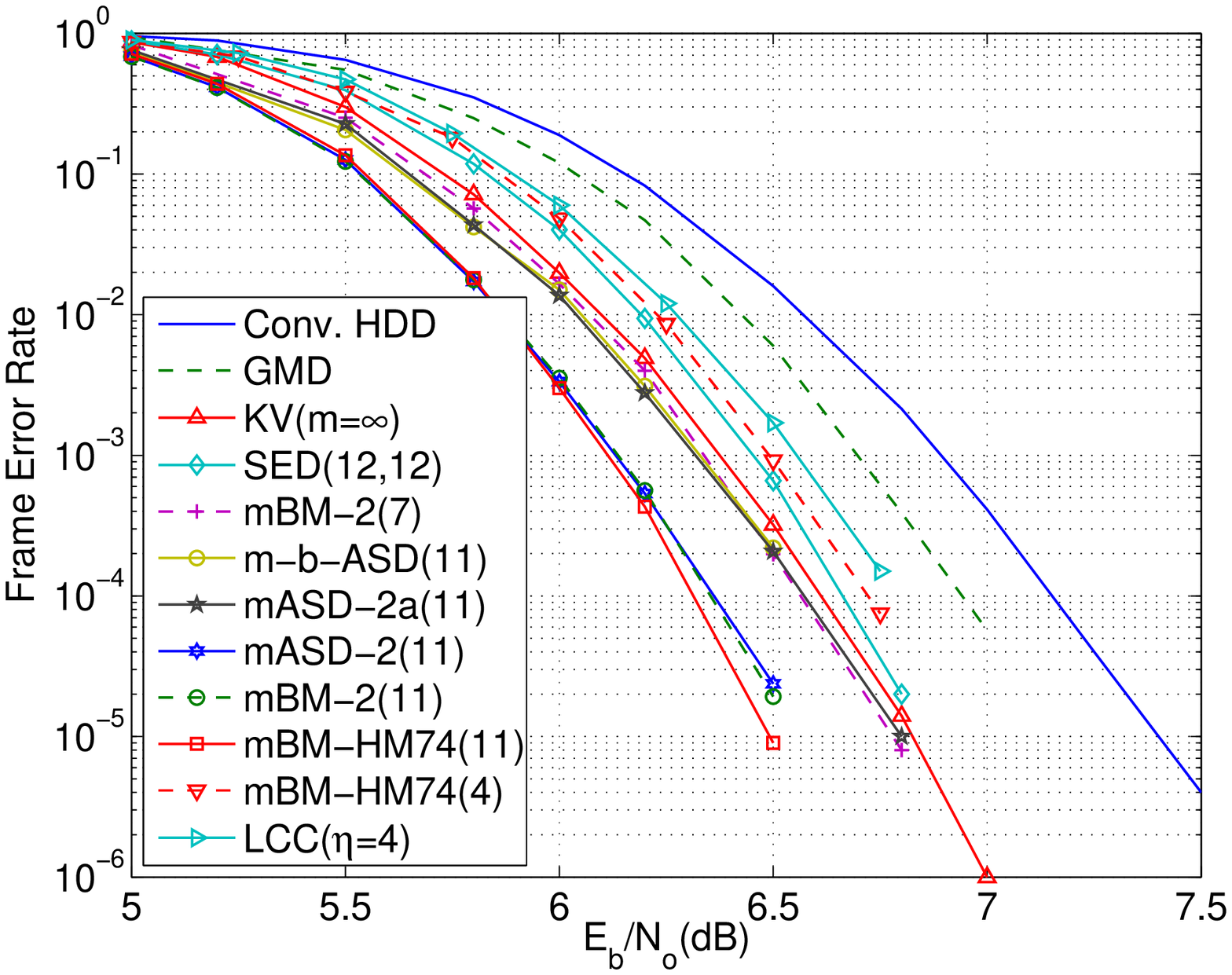}\caption{\label{fig:simulation} Performance of various decoding algorithms
for the (255,239) RS code over an AWGN channel.}
\end{minipage}
\end{figure}

Using simulations, we consider the performance of the (255,239) RS
code over an AWGN channel with BPSK as the modulation format. The
mBM-1 curve corresponds to standard error-and-erasure BM decoding
with multiple erasure patterns. For $l>1$, the mBM-$l$ curves correspond
to error-and-erasure BM decoding with multiple decoding trials using
both erasures and top-$l$ symbols. The mASD-$m$ curves correspond
to multiple ASD decoding trials with maximum multiplicity $m$. The
number of trial decoding patterns is $2^{R}$ where $R$ is denoted
in parentheses in each algorithm's acronym (e.g., m-BM-2(11) uses
$R=11$).

Fig. \ref{fig:rdcurve} shows the R-D curves for various algorithms
at~$E_{b}/N_{o}=5.2$~dB. For this code, the fixed threshold for
decoding is $D=n-k+1=17$. Therefore, one might expect that algorithms
whose average distortion is less than 17 should have a frame error
rate (FER) less than $\frac{1}{2}$. The R-D curve allows one to estimate
the number of decoding patterns required to achieve this FER. Conventional
BM decoding is very similar to mBM-1 decoding at rate 0. Notice that
the mBM-1 algorithm at rate 0, which is very similar to conventional
BM decoding, has an expected distortion of roughly 24. For this reason,
the FER on conventional decoding is close to 1. The R-D curve tells
us that trying roughly $2^{16}$ (i.e., $R=16$) erasure patterns
would reduce the FER to roughly $\frac{1}{2}$ because this is where
the distortion drops down to 17. Likewise, the mBM-2(11) algorithm
has an expected distortion of less than 14. So we expect (and our
simulations confirm) that the FER should be less than $\frac{1}{2}$.
One weakness of this approach is that the R-D describes only the average
distortion and does not directly consider the probability that the
distortion is greater than 17. Still, we can make the following observations
from the R-D curve. Even at low rates (e.g., $R\ge4$), we see that
the distortion~$D$~achieved by mBM-2 is roughly the same as mBM-3,
mASD-2, and mASD-3 but smaller than mASD-2a and mBM-1. This implies
that mBM-2 is no worse than the more complicated ASD based approaches
for a wide range of rates (i.e., $4\leq R\leq35$). 

The FER of various algorithms can be seen in Fig. \ref{fig:simulation}.
The focus on $R=11$ allows us to make fair comparisons with SED(12,12).
With the same number of decoding trials, mBM-2(11) outperforms SED(12,12)
by 0.3 dB at an FER$=10^{-4}$. Even mBM-2(7), with many fewer decoding
trials, outperforms both SED(12,12) and the KV algorithm with~$m=\infty$.
Among all our proposed algorithms with rate~$R=11$, the mBM-HM74(11)
achieves the best performance. This algorithm uses the Hamming (7,4)
covering code for the 7 LRPs and the R-D approach for the remaining
codeword positions. Meanwhile, small differences in the performance
between mBM-2(11), mBM-3(11), mASD-2(11), and mASD-3(11) suggest that:
(i) taking care of the~$2$~most likely symbols at each codeword
position is good enough for multiple decoding of high-rate RS code
and (ii) multiple runs of error-and-erasure decoding is almost as
good as multiple runs of ASD decoding. Recall that this result is
also correctly predicted by the R-D analysis. Moreover, it is quite
reasonable since we know that the gain of GS decoding, with infinite
multiplicity, over the BM algorithm is negligible for high-rate RS
codes. To compare with the LCC($\eta=4$) Chase-type approach used
in \cite{Bellorado-isit06}, we also consider the mBM-HM74(4) algorithm,
which uses the Hamming (7,4) covering codes for the 7 LRPs and the
hard decision pattern for the remaining codeword positions. This shows
that the covering code achieves better performance with the same number
($2^{4}$) decoding attempts. The comparison is not entirely fair,
however, because of their low-complexity approach to multiple decoding.
We believe, nevertheless, that their technique can be generalized
to covering codes.

\section{Conclusion\label{sec:Conclusion}}

A rate-distortion approach is proposed as a unified framework to analyze
multiple decoding trials, with various algorithms, of RS codes in
terms of performance and complexity. A connection is made between
the complexity and performance (in some asymptotic sense) of these
multiple-decoding algorithms and the rate and distortion of an associated
R-D problem. Covering codes are also combined with the rate-distortion
approach to mitigate the suboptimality of random codes when the effective
block-length is not large. As part of this analysis, we also present
numerical and analytical computations of the rate-distortion function
for sequences of independent but non-identical sources. Finally, the
simulation results show that our proposed algorithms based on the
R-D approach achieve a better performance-versus-complexity trade-off
than previously proposed algorithms. One key result is that, for high-rate
RS codes, multiple-decoding using the standard BM algorithm is as
good as multiple-decoding using more complex ASD algorithms.

In this paper, we only discuss the rate-distortion approach to solve
the problem in (\ref{eq:ProbStatement}). However, the performance
can be further improved by focusing on the rate-distortion error-exponent.
This allows us to approximately solve the covering problem for finite~$n$
rather than just as~$n\rightarrow\infty$. The complexity of multiple
decoding can also be decreased by using clever techniques to lower
the complexity per decoding trial (e.g., \cite{Bellorado-isit06}).
We will address these two improvements in a future paper.

\appendices

\section{Proof of Theorem \ref{thm:(Factored-Blahut)}\label{sec:AppBlahut}}

First, let us recall that for each source component~$x_{i}$, the
B-A algorithm computes the R-D pair in the following steps:
\begin{enumerate}
\item Choose an arbitrary all-positive test-channel input-probability distribution
vector~$\underline{q}^{(0)}$.
\item Iterate the following steps at~$t=1,2,\ldots$\[
Q_{k_{i}|j_{i}}^{(t)}=\frac{q_{k_{i}}^{(t)}\exp(s\delta_{j_{i}k_{i}})}{\sum_{k_{i}}q_{k_{i}}^{(t)}\exp(s\delta_{j_{i}k_{i}})}\,\,\,\mbox{and\ensuremath{\,\,\,}}\sum_{j_{i}}p_{j_{i}}Q_{k_{j}|j_{i}}^{(t)}=q_{k_{i}}^{(t+1)}\]
where~$Q_{k_{i}|j_{i}}=\Pr(\hat{x}_{i}=k_{i}|x_{i}=j_{i})$~is the
transition probability. It is shown by B-A that~$q_{k_{i}}^{(t)}\rightarrow q_{k_{i}}^{\star}$~and~$Q_{k_{i}|j_{i}}^{(t)}\rightarrow Q_{k_{i}|j_{i}}^{\star}$as~$t\rightarrow\infty$.
\end{enumerate}
The rate and distortion can be computed by~$R_{i,s}=\sum_{j_{i}}\sum_{k_{i}}p_{j_{i}}Q_{k_{i}|j_{i}}^{\star}\log\frac{Q_{k_{i}|j_{i}}^{\star}}{\sum_{j_{i}}p_{j_{i}}Q_{k_{i}|j_{i}}^{\star}}$~and~$\sum_{i=1}^{n}D_{i,s}=\sum_{j_{i}}\sum_{k_{i}}p_{j_{i}}Q_{k_{i}|j_{i}}^{\star}\rho_{j_{i}k_{i}}$

Now, we will prove Theorem \ref{thm:(Factored-Blahut)}. Since the
input-distribution vector of the test channel is an arbitrary all-positive
vector, we choose~$q_{K}^{(0)}$~so that it can be factorized as
follows~$q_{K}^{(0)}=\prod_{i=1}^{n}q_{k_{i}}^{(0)}$ .

Suppose after step~$t$, we have~$q_{K}^{(t)}=\prod_{i=1}^{n}q_{k_{i}}^{(t)}$~then
by the iterative computing process in the B-A algorithm we have \[
Q_{K|J}^{(t)}=\frac{q_{K}^{(t)}\exp(s\delta_{JK})}{\sum_{K}q_{K}^{(t)}\exp(s\delta_{JK})}=\frac{\prod_{i=1}^{n}q_{k_{i}}^{(t)}\exp(s\delta_{j_{i}k_{i}})}{\sum_{k_{1}}\ldots\sum_{k_{n}}\prod_{i=1}^{n}q_{k_{i}}^{(t)}\exp(s\delta_{j_{i}k_{i}})}=\prod_{i=1}^{n}\frac{q_{k_{i}}^{(t)}\exp(s\delta_{j_{i}k_{i}})}{\sum_{k_{i}}q_{k_{i}}^{(t)}\exp(s\delta_{j_{i}k_{i}})}=\prod_{i=1}^{n}Q_{k_{i}|j_{i}}^{(t)}\]

\[
q_{K}^{(t+1)}=\sum_{J}p_{J}Q_{K|J}^{(t)}=\sum_{j_{1}}\ldots\sum_{j_{n}}\prod_{i=1}^{n}p_{j_{i}}Q_{k_{i}|j_{i}}^{(t)}=\prod_{i=1}^{n}\sum_{j_{i}}p_{j_{i}}Q_{k_{j}|j_{i}}^{(t)}=\prod_{i=1}^{n}q_{k_{i}}^{(t+1)}\]

where~$p_{J}=\Pr(x^{n}=J)=\prod_{i=1}^{n}\Pr(x_{i}=j_{i})=\prod_{i=1}^{n}p_{j_{i}}$~follows
from the independence of source components.

Hence, by induction we know the factorization rule still applies after
every step~$t$~in the process. Therefore, we have\[
q_{K}^{\star}=\lim_{t\rightarrow\infty}q_{K}^{(t)}=\lim_{t\rightarrow\infty}\prod_{i=1}^{n}q_{k_{i}}^{(t)}=\prod_{i=1}^{n}\lim_{t\rightarrow\infty}q_{k_{i}}^{(t)}=\prod_{i=1}^{n}q_{k_{i}}^{\star}\,\,\mbox{\ensuremath{\,\,} and}\,\,\,\, Q_{K|J}^{\star}=\lim_{t\rightarrow\infty}Q_{K|J}^{(t)}=\lim_{t\rightarrow\infty}\prod_{i=1}^{n}Q_{k_{i}|j_{i}}^{(t)}=\prod_{i=1}^{n}\lim_{t\rightarrow\infty}Q_{k_{i}|j_{i}}^{(t)}=\prod_{i=1}^{n}Q_{k_{i}|j_{i}}^{\star}\]

and then we can show that\[
R_{s}=\sum_{J}\sum_{K}p_{J}Q_{K|J}^{\star}\log\frac{Q_{K|J}^{\star}}{\sum_{J}p_{J}Q_{K|J}^{\star}}=\sum_{j_{1},\ldots,j_{n}}\sum_{k_{1},\ldots,k_{n}}\prod_{i=1}^{n}p_{j_{i}}Q_{k_{i}|j_{i}}^{\star}\sum_{i=1}^{n}\log\frac{Q_{k_{i}|j_{i}}^{\star}}{\sum_{j_{i}}p_{j_{i}}Q_{k_{i}|j_{i}}^{\star}}=\sum_{i=1}^{n}\sum_{j_{i}}\sum_{k_{i}}p_{j_{i}}Q_{k_{i}|j_{i}}^{\star}\log\frac{Q_{k_{i}|j_{i}}^{\star}}{\sum_{j_{i}}p_{j_{i}}Q_{k_{i}|j_{i}}^{\star}}=\sum_{i=1}^{n}R_{i,s}\]

\[
D_{s}=\sum_{J}\sum_{K}p_{J}Q_{K|J}^{\star}\rho_{JK}=\sum_{j_{1},\ldots,j_{n}}\sum_{k_{1},\ldots,k_{n}}\prod_{i=1}^{n}p_{j_{i}}Q_{k_{i}|j_{i}}^{\star}\sum_{i=1}^{n}\rho_{j_{i}k_{i}}=\sum_{i=1}^{n}\sum_{j_{i}}\sum_{k_{i}}p_{j_{i}}Q_{k_{i}|j_{i}}^{\star}\rho_{j_{i}k_{i}}=\sum_{i=1}^{n}D_{i,s}.\]

\bibliographystyle{IEEEtran}

\end{document}